\documentclass[11pt]{article}

\usepackage{amsmath,amssymb}
\usepackage[utf8]{inputenc}
\usepackage{color}
\usepackage{graphicx}
\usepackage{subcaption}

\newtheorem{teo}{Theorem}[section]
\newtheorem{lem}{Lemma}[section]

\newtheorem{prop}{Proposition}[section]

\newtheorem{remark}{Remark}[section]

\newenvironment{proof}{{\bf Proof.}}{\hfill $\Box $  \newline}

\title{Dirac cones for bi- and trilayer Bernal-stacked graphene in a quantum graph model}

\author{\\ C\'esar R. de Oliveira and Vin{\'{\i}}cius L. Rocha\\ \\ \small{Departamento de Matem\'atica, UFSCar, S\~ao Carlos, SP, 13560-970 Brazil}}
\date{}

\begin{document}

\maketitle

\begin{abstract}  
A quantum graph model for a single sheet of graphene is extended to bilayer and trilayer Bernal-stacked graphene; the spectra are characterized and the dispersion relations explicitly obtained; Dirac cones are then proven to be present only for trilayer graphene, although the bilayer has a gapless parabolic band component. Our model rigorously exhibits  basic facts from tight-binding calculations,  effective two-dimensional models and a $\pi$-orbital continuum model with nearest-neighbour tunneling that have been discussed in the physics literature.
\end{abstract}

\

\


\


\

\section{Introduction}

Graphene is a single sheet of graphite and, as a ``true'' two-dimensional material with a peculiar (carbon honeycomb) periodic configuration, has outstanding physical properties~\cite{katsnelson,dasSarmaEtAl,castroEtAl}. Perhaps the most important of such properties is the presence of {\em Dirac cones} located at a finite number of points in the Brillouin zone, the so-called {\em D-points} (or Dirac points). It is accepted that the  motion of an electron spectrally near each D-point is approximately described by a two-dimensional Dirac operator with effective zero mass and effective speed of light  $c/300$;  such exotic framework is responsible for many particular physical properties. The interest in graphene has substantially increased after it was experimentally isolated in 2004.

There are many relevant works in the physics literature describing spectral properties of graphene (even before its controlled isolation), and they are based on approximations as the tight-binding one, density functional calculations and numerical simulations (in particular {\em ab initio} calculations). The first tight-binding study goes back to 1947 and was done by Wallace~\cite{wallace} (see also~\cite{ACM1, coulson1}), who has found that graphene is a zero-gap semiconductor (sometimes characterized as a semimetal) with a linear dispersion relation (Dirac cone),  which defines the D-points.   Let $\theta=(\theta_1,\theta_2)$ denote the quasimomentum in the first Brillouin zone $\mathcal{B}:=[-\pi,\pi]^{2}$,  and $\lambda(\theta)$ the associated dispersion relation; roughly,  $\theta_K\in \mathcal B$ is a D-point candidate if there is a constant $\gamma\ne0$ so that  
\[
\lambda(\theta)-\lambda(\theta_K)\approx\pm \gamma|\theta-\theta_K|,
\]
and  we have a Dirac cone, since the valence (the~``-'' sign above) and the conducting (the~``+'' sign) bands touch (approximately) linearly. We present a precise definition in Section~\ref{sectionDiracCones}. 

See~\cite{PGLMP,jacqmin,ozawaetall2019,DR2019} for descriptions of interesting experiments with {\em artificial graphene,} that is, some synthetic structures  that permit the study of photonic crystals with Dirac cone dispersion and topologically protected edge states; the idea is to built systems for which the physics is simpler to explore than graphene itself, and  electrons may be replaced with  photons, plasmons or  microcavity polaritons.

There are few mathematical results related to the presence of Dirac cones in models of graphene. We underline two interesting works with different approaches, one by Fefferman and Weinstein~\cite{FW1} and another by Kuchment and Post~\cite{KP1}. In~\cite{FW1} the authors consider a two-dimensional Schr\"odinger operator with smooth potentials, which are periodic with respect to the honeycomb lattice and satisfy some symmetry conditions, and have proven the presence of Dirac cones as well as stability under suitable perturbations. In~\cite{KP1} the authors consider a (honeycomb) quantum graph model for  graphene and derive its spectral properties, in particular the presence of Dirac cones. An approach through quantum graphs seems to have been first proposed by Linus Pauling~\cite{pauling} to describe some chemical systems and implemented in~\cite{RScherr}. In summary, in this quantum graph model the electron motion is restricted to the edges and vertices of the honeycomb graph;  Sturm-Liouville operators on edges are considered with ``natural'' boundary conditions at the vertices; an advantage of this model is that, by applying the Floquet-Bloch theory~\cite{eastham,Floquettheory3RS}, it was possible to describe details of the spectrum of the graphene Hamiltonian and its dispersion relation can be explicitly computed.

There are also theoretical and experimental interests in systems composed of finite layers of graphene (see, for instance \cite{BYWCAAZABA2, camposEtAl, latilH, LJXABA1, MCAVFAB2, mccannKosh, MMABAandABAB,  partoensPeeters,    paton, RSRFAB1}), also as an approximation for the bulk graphite. There are two important remarks here; first, the strength of the bond between consecutive layers is much weaker than the bonds between neighbour carbon atoms in the same layer; second, there are different possibilities for stacking layers of graphene, and it was experimentally found  that physical properties depend on how layers are stacked. 

The purpose of this work is to investigate spectral properties and the possible presence of Dirac cones for bilayer and trilayer graphene systems, always Bernal-stacked (also called AB-stacked) and modelled by quantum graphs, that is, we extend the mathematical analysis of~\cite{KP1} to two and three  AB-stacked  graphene sheets. In this setting, we need that all edges are of equal lengths, so the weak interaction between layers should be modelled in an alternative way. We propose to model the weak interaction between layers through a ``small''  weight parameter~$t_0$ in the sum of derivatives of wavefunctions at each vertex and also regulating the values of functions at vertices ($t_0$ is the same for all vertices with connections to other layers).   Details of the proposed model appear in Section~\ref{nAstackednlayergraphene}.

We note that the Bernal-stacked form of multilayer graphene is the most stable one; ahead we describe our proposal of how to implement it in this graph model. Another possibility is the so-called AA-stacked, in which all graphene sheets and carbon-carbon connections between consecutive sheets are alined. In another work~\cite{deORochaAA}, we discuss this possibility for many layers. 

 Our main results here may be summarized as follows. 

\begin{enumerate}
\item We have proven that the bilayer graphene model has no Dirac cones, whereas such cones are present in trilayer graphene dispersion relation. However, the bilayer dispersion relation is gapless and with quadratic touch.
\item For both bilayer and trilayer cases, the spectra have eigenvalues of infinite multiplicity (the eigenvalues of the Dirichlet Hamiltonian in a single edge) and an absolutely continuous component built of closed intervals (bands). The singular continuous spectrum is always absent.
\item We compare our findings to some results obtained in the physics literature.
\item Although the  interlayer interaction parameter $t_0$ should be small to model the multilayer graphene, our results here hold for all $t_0>0$. On the other hand, in our study of multilayer AA-stacked graphene~\cite{deORochaAA}, the larger the number of layers the smaller~$t_0$ is required.
\end{enumerate}

In Section~\ref{nAstackednlayergraphene} we introduce the quantum graph geometry and the proposed Schr\"odinger operator for the Bernal-stacked graphene. In Section~\ref{sectionSpectrumGnA}, we perform a spectral analysis of the bilayer and trilayer graphene Schr\"odinger operators. The existence of Dirac cones is discussed in Section~\ref{sectionDiracCones} and some comparisons with the physics literature are done in Section~\ref{sectComparePL}.

\section{Multilayer Bernal-stacked graphene}\label{nAstackednlayergraphene}

 The hexagonal single-layer graphene lattice~$G$ consists of two triangular sublattices, $g_{\mathbf A}$ and $g_{\mathbf B}$, defined by
\begin{equation}
g_{\mathbf A} = \mathbf A + g \quad\text{and}\quad g_{\mathbf B} = \mathbf B + g,
\end{equation}
where $\mathbf A=(0,0)$, $\mathbf B=(1,0)$ and $g$ is the triangular lattice $g:=\mathbb{Z} {\mathbf E_1} \oplus\mathbb{Z} {\mathbf E_2}$, with ${\mathbf E_1}=(0,\sqrt{3})$ and ${\mathbf E_2}=(3/2,\sqrt{3}/2)$ denoting the \textit{lattice vectors}. The elements of $g_{\mathbf A}$ and $g_{\mathbf B}$ are said \textit{type-A} and \textit{type-B vertices}, respectively. The hexagonal 2D lattice  is given by $G=g_{\mathbf A}\cup g_{\mathbf B}$. It is supposed that carbon atoms are located at the vertices of $G$ and the covalent bonds are represented by edges of length~1, as shown in Figure \ref{FigureGrafeno}.

\begin{figure}[h]
	
	\centering 
	\includegraphics[width=7cm]{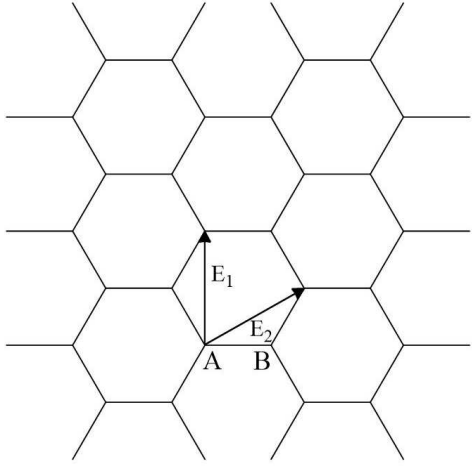}
	\caption{The hexagonal 2D lattice $G$ and its lattice vectors ${\mathbf E_1}$ and ${\mathbf E_2}$.}
	\label{FigureGrafeno}
	
\end{figure} 

Now we introduce the periodic quantum graphs that is proposed to represent the \textit{AB-stacked bilayer/trilayer graphene}; we need a metric graph structure and a suitable Hamiltonian.

The structure of the bilayer graphene, denoted by $\mathcal{G}_{2}$, consists of two sheets of graphene, $G_1$ and $G_2$, each one as defined above, stacked such that a type-A vertex of $G_1$ is located exactly above the corresponding type-A vertex of $G_2$ and alike for type-B vertices; that is, the sheets are alined. The only connections are that  each type-A vertex of~$G_1$ is connected to the nearest three type-B vertices of~$G_2$ (see Figure \ref{GraphBilayer}); such connections are performed through additional edges, and since all edges must have the same length~\cite{KP1} (see, for instance, the representations~\eqref{representationoffunctionsABonW} and~\eqref{representation of functions on W n=3} ahead) the distance between layers is such that the length of such additional edges is~1 (the length of graphene edges). 

The trilayer graphene, denoted by $\mathcal{G}_3$, consists of three graphene sheets, $G_1,G_2$ and $G_3$, in an analogous alined stacking of the bilayer case, but now each type-B vertex of $G_2$ is linked up (only) with the nearest  three type-A vertices of~$G_1$ and the nearest  three type-A vertices of~$G_3$, a total of six additional connections (see Figure \ref{GraphTrilayer}). In both cases, the distance between two consecutive graphene sheets is taken in such way that every edge in $\mathcal{G}_{n}$ have length~1, for~$n=2,3$; ahead we describe how we control the interaction intensity between consecutive graphene sheets.

\begin{figure}
	\centering
	\begin{subfigure}{.5\textwidth}
		\centering
		\includegraphics[width=1.0\linewidth]{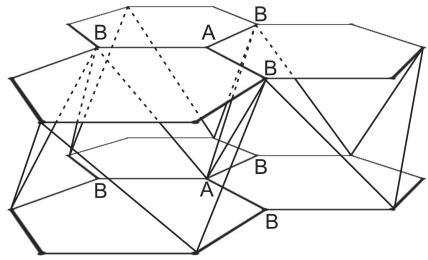}
		\caption{Bilayer graphene.}
		\label{GraphBilayer}
	\end{subfigure}
	\begin{subfigure}{.5\textwidth}
		\centering
		\includegraphics[width=.9\linewidth]{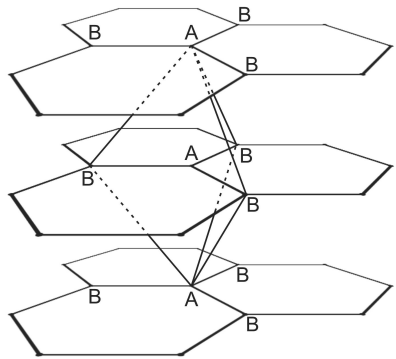}
		\caption{Trilayer graphene.}
		\label{GraphTrilayer}
	\end{subfigure}
	\caption{The lattice structures of the AB-stacked multilayer graphene. Some type-A and type-B points are labelled. For simplicity, in~(b) we present only a few links between sheets.}
	\label{GraphTrilayerBoth}
\end{figure}

Consider the action of the group $\mathbb{Z}^2$ on  $\mathcal{G}_{n}$,
\begin{equation}\label{actionshift}
\mathcal{S}:\mathbb{Z}^2 \times  \mathcal{G}_{n} \rightarrow \mathcal{G}_{n}, \quad \mathcal{S}(\mathbf p,\mathbf x):=p_{1}{\mathbf E_1}+p_{2}{\mathbf E_2}+\mathbf x,
\end{equation}
that is, $\mathcal{S}$ shifts $\mathbf x\in\mathcal{G}_{n}$ by $p_{1}{\mathbf E_1}+p_{2}{\mathbf E_2}$, $\mathbf p=(p_1,p_2)\in\mathbb{Z}^{2}$. As fundamental domain of $\mathcal{S}$, we choose the set $\mathcal{W}_{n}$ as shown in Figure \ref{FundamentalDomainBoth}, which contains two points and three edges of each graphene sheet $G_k$, and the edges that connect consecutive layers (as discussed above). It will be convenient to direct the edges as in  Figure \ref{FundamentalDomainBoth}. 

Our proposal is a balance between physics and explicitly calculations. For instance, in a more realistic two-layer AB geometry, a layer is shifted by one-edge length with respect to the other, and there are B vertices on top of (some) A vertices, but this is harder to implement here. Then we have thought of simulating the true geometry by keeping the layers aligned and connecting lower A vertices to upper B vertices, but in such configuration each lower A vertex has three nearest B neighbours, and connecting each lower A vertex to all its upper B nearest neighbours  includes some symmetry that have allowed us to perform calculations.  Furthermore, since the distance to the closest lying A atoms of both layers is smaller than~1, such coupling may not be implemented. Of course, the obtained results are the final justification of the model.

For $n=2,3$, let $E(\mathcal{G}_{n})$ denote the set of edges of $\mathcal{G}_{n}$, $EL(\mathcal{G}_{n})$ denote the set of edges in the graphene {\bf layers} (e.g., edges in $G_1$ and~$G_2$ in the case of two layers) and $EC(\mathcal{G}_{n})$  the set of   edges {\bf connecting} consecutive sheets of graphene. We will label the edges in ${EL}(\mathcal{G}_{n})$ by the letter~``$\mathbf{a}$" whereas the edges in ${EC}(\mathcal{G}_{n})$ will be labeled by~``$\mathbf{f}$". Note that the sets ${EL}(\mathcal{G}_{n})$ and ${EC}(\mathcal{G}_{n})$ are disjoint and 
\[
E(\mathcal{G}_{n})={EL}(\mathcal{G}_{n})\cup{EC}(\mathcal{G}_{n}).
\] For instance, in Figure \ref{FundamentalDomainTrilayer} we see that $\mathbf a_{11},\mathbf a_{13},\mathbf a_{21}$ are in ${EL}(\mathcal{G}_{3})$ whereas $\mathbf f_1,\mathbf f_2$ are in ${EC}(\mathcal{G}_{3})$. This distinction between different types of edges will be important ahead, in particular in the boundary conditions~\eqref{continuity} and~\eqref{vanishingflux}.

Given a vertex~$\mathbf v$, denote by $E_{\mathbf v}(\mathcal{G}_{n})$  the set of edges of $\mathcal{G}_{n}$ that are connected to~$\mathbf v$, and by ${EL_{\mathbf v}}(\mathcal{G}_{n})$ and ${EC_{\mathbf v}}(\mathcal{G}_{n})$ the elements  of ${EL}(\mathcal{G}_{n})$ and ${EC}(\mathcal{G}_{n})$ connecting to~$\mathbf v$, respectively (note that, for some vertices, ${EC_{\mathbf v}}(\mathcal{G}_{n})$ may be the empty set).

\begin{figure}
	\centering
	\begin{subfigure}{.5\textwidth}
		\centering
		\includegraphics[width=.9\linewidth]{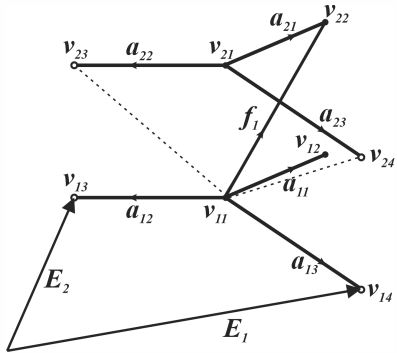}
		\caption{$\mathcal{W}_2$ of the bilayer graphene.}
		\label{FundamentalDomainBilayer}
	\end{subfigure}
	\begin{subfigure}{.5\textwidth}
		\centering
		\includegraphics[width=.9\linewidth]{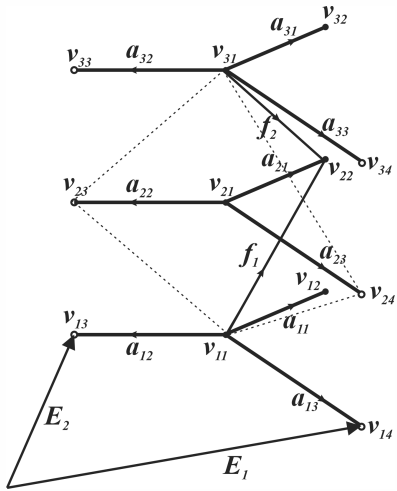}
		\caption{$\mathcal{W}_3$ of the trilayer graphene.}
		\label{FundamentalDomainTrilayer}
	\end{subfigure}
	\caption{The fundamental domains $\mathcal{W}_{2}$ and $\mathcal{W}_{3}$.}
	\label{FundamentalDomainBoth}
\end{figure}

Since $\mathcal{G}_{n}$ is supposed to be embedded into the Euclidean space~$\mathbb R^3$,  we  identify each edge $\mathbf e\in E(\mathcal{G}_{n})$ with the segment $[0,1]$, which identifies the end points of~$\mathbf e$ with $0$ and $1$. One can naturally define the Hilbert space of all square integrable functions on $\mathcal{G}_{n}$,
\begin{equation*}
{\mathrm L}^{2}(\mathcal{G}_{n})=\bigoplus_{\mathbf e\in E(\mathcal{G}_{n})}{\mathrm L}^{2}(\mathbf e). 
\end{equation*} 
We denote functions on $\mathcal{G}_{n}$ by $u=\{u_{\mathbf e}\}_{e\in E(\mathcal{G}_{n})}$, where $u_{\mathbf e}$ is a function defined on the edge~$\mathbf e$.

To see $\mathcal{G}_{n}$ as a periodic quantum graph, it remains to introduce the Schr\"odinger operator acting in ${\mathrm L}^{2}(\mathcal{G}_{n})$. Let $q_0:[0,1]\rightarrow \mathbb{R}$ be a real-valued even continuous function, that is, $q_0(x)=q_0(1-x)$ for all $x\in [0,1]$. As we have identified the edges of $\mathcal{G}_{n}$ with the segment $[0,1]$, we can define a \textit{potential} $q=\{q_{\mathbf e}\}_{\mathbf e\in E(\mathcal{G}_{n})}$ on $\mathcal{G}_{n}$, where $q_{\mathbf e}=q_0$, for all $\mathbf e\in E(\mathcal{G}_{n})$. Note that due to the evenness assumption on $q_0$, the potential $q$ does not depend on the orientations chosen along the edges; an even potential is a consequence of the fact that we have the same kind of atoms at the vertices. We have the following result, which is analogous to Proposition~2.1 in~\cite{DoKuch}.

\begin{lem}\label{lemapotencial}
	The above potential $q$  is invariant with respect to the symmetry group of the $\mathcal{G}_{n}$. 
\end{lem}

Finally, we  introduce the  \textit{Bernal-stacked multilayer graphene Schr\"odinger operator} $H_{n}$, $n=2,3$, which acts on a function $u\in D(H_{n})$ as
\begin{equation}\label{defHnA} 
(H_{n}u_{\mathbf e})(x_{\mathbf e}):= (-\Delta+q_{\mathbf e}(x_{\mathbf e}))u_{\mathbf e}(x_{\mathbf e}) = -\frac{{\mathrm d}^{2}u_{\mathbf e}(x_{\mathbf e})}{{\mathrm d}x^2}+q_{\mathbf e}(x_{\mathbf e})u_{\mathbf e}(x_{\mathbf e})\,,
\end{equation}
for each $\mathbf e\in E(\mathcal{G}_{n})$. When the context is  clear, the subscript~``$\mathbf e$" will be omitted. The domain $D(H_{n})$ consists of the functions $u$ on $\mathcal{G}_{n}$ that satisfy the following four set of conditions; $0<t_0\le1$ is a  parameter regulating the interaction between consecutive layers.
\begin{itemize}
	\item[(i)] $u_{\mathbf e}\in H^{2}(\mathbf e)$, for all $\mathbf e\in E(\mathcal{G}_{n})$, where $H^{2}(\mathbf e)$ is the usual \textit{Sobolev Space} in the edge~$\mathbf e$; so the functions $u_{\mathbf e}$ and their derivatives $u_{\mathbf e}'$ are (bounded) continuous on each edge;
	\item[(ii)] $\displaystyle\sum_{\mathbf e\in E(\mathcal{G}_{n})}\|u_{\mathbf e}\|_{H^{2}(\mathbf e)}^{2}<\infty$; 
	\item[(iii)] The  \textit{weighted continuity condition}, which requires that, at each vertex~$v\in V(\mathcal{G}_{n})$, 
		\begin{equation}\label{continuity}
	u_{\mathbf{a}_1}(\mathbf v)=u_{\mathbf{a}_2}(\mathbf v)=\frac{u_{\mathbf{f}_1}(\mathbf v)}{t_0}=\frac{u_{\mathbf{f}_2}(\mathbf v)}{t_0},
	\end{equation}
for all $\mathbf a_1,\mathbf a_2\in {EL_{\mathbf v}}(\mathcal{G}_{n})$ and all $\mathbf f_1,\mathbf f_2\in {EC_{\mathbf v}}(\mathcal{G}_{n})$ (these sets are described above).  
	\item[(iv)] The \textit{weighted Kirchhoff condition}, which requires that, at each  vertex~$\mathbf v\in V(\mathcal{G}_{n})$, 
	\begin{equation}\label{vanishingflux}
		\displaystyle\sum_{\mathbf a\in {EL_{\mathbf v}}(\mathcal{G}_{n})}u_{\mathbf a}^{\prime}(\mathbf v)+\displaystyle\sum_{\mathbf f\in {EC_{\mathbf v}}(\mathcal{G}_{n})}t_0\,u_{\mathbf f}^{\prime}(\mathbf v)=0,
	\end{equation} 
	where $u_{\mathbf e}^{\prime}(\mathbf v)$ is the derivate of $u_{\mathbf e}$ directed from $\mathbf v$ to the other vertex connected to~$\mathbf e$.
\end{itemize}

\begin{remark}
A word on the role of the interaction parameter~$t_0$. The interaction between consecutive layers of graphene is weaker than the interaction between neighbouring carbon atoms in the same layer. If $t_0=1$ we would have no distinction between the interactions; so we have proposed to take $0<t_0<1$ as a way to control the influence of an atom in a different sheet; by~\eqref{vanishingflux}, the smaller~$t_0$ the smaller the influence of such edge to the ``flux balance'' and, by~\eqref{continuity}, the smaller the value of the corresponding function to compensate the $t_0$ in the denominator. However, the general results here (Theorems~\ref{TheoChacSpectrum} and~\ref{mainresult}) hold true for all positive values of the interaction parameter~$t_0>0$, and so it is natural to ask whether such property is shared with other models, as the tight-binding one; we don't know the answer. Note that $t_0=0$ is a singular limit that  is not employed here (although, intuitively,  we get a decoupling among graphene sheets in this limit). We would like to emphasize that we are here proposing the use of~$t_0$ for modelling the weak interactions between graphene sheets because, in the  literature on quantum graphs, it is not clear how to simulate such interactions. 
\end{remark}

	This definition makes $H_{n}$ an unbounded  self-adjoint operator \cite{K1,KS,KBlivro} and, by the evenness condition on the potential and Lemma~\ref{lemapotencial}, it is invariant with respect to all symmetries of the graph $\mathcal{G}_{n}$.

\begin{remark}(Self-Adjointness of $H_n$)
	 By Theorems~1.4.4 and~1.4.11 in~\cite{KBlivro}, in order to prove that the Schr\"odinger operator $H_n$, given by \eqref{defHnA}-\eqref{continuity}-\eqref{vanishingflux}, is self-adjoint, it is necessary and sufficient that there exist two matrices ${\mathbf{A_v}}$ and $\mathbf {B_v}$, for each vertex $\mathbf  v\in V(\mathcal{G}_n)$, both with order $ \mathbf{d_v} \times \mathbf{d_v}$, where $ \mathbf{d_v}$ is the \textit{degree} of~$\mathbf v$ (that is, the number of edges connected to the vertex~$\mathbf v$), such that: 
	\begin{enumerate}
		\item[(j)] the $ \mathbf{d_v} \times 2 \mathbf{d_v}$ matrix $\begin{bmatrix} {\mathbf{A_v}} & {\mathbf{B_v}}	\end{bmatrix}$ has maximal rank;
		\item[(jj)] the matrix ${\mathbf{A_v}}{\mathbf{B_v}}^*$ is self-adjoint, where ${\mathbf{B_v}}^*$ is the adjoint of ${\mathbf{B_v}}$;
		\item[(jjj)] ${\mathbf{A_v}}\mathbf F(\mathbf v)={\mathbf{B_{\mathbf v}}} \mathbf F^\prime(\mathbf{v})$, 
		where the vector $\mathbf{F(\mathbf v)}$ and $\mathbf{F}^\prime(\mathbf v)$ are defined by
		\begin{equation*}\label{Fvector}
		\mathbf F(\mathbf v):=\begin{bmatrix}u_{\mathbf{e}_1}(\mathbf v),\ldots,u_{\mathbf{e}_{d_{\mathbf v}}}(\mathbf v)\end{bmatrix}^\intercal
		\end{equation*}
		and
		\begin{equation*}\label{FPrimevector}
		\mathbf F^\prime(\mathbf v):=\begin{bmatrix}u_{\mathbf{e}_1}^\prime(\mathbf v),\ldots,u_{\mathbf{e}_{d_{\mathbf v}}}^\prime(\mathbf v)\end{bmatrix}^\intercal,
		\end{equation*}		
		with $\mathbf e_1,\ldots,\mathbf e_{d_{\mathbf v}}\in E_{\mathbf v}(\mathcal{G}_n)$.
	\end{enumerate}
	Let $n=2$ and consider the vertex $\mathbf v_{11}\in V(\mathcal{G}_2)$, which has degree~6 (see Figure \ref{FundamentalDomainBilayer}).  Then the $6\times 6$-matrices $\mathbf A_{\mathbf v_{11}}$ and $\mathbf B_{\mathbf v_{11}}$ are given by 
	\begin{equation*}	
	\mathbf A_{\mathbf v_{11}}:=\begin{bmatrix}
	1 & -1 & 0 & 0 & 0 & 0 \\
	0 & 1 & -1 & 0& 0 & 0 \\
	0 & 0 & t_0 & -1 & 0 & 0 \\
	0 & 0 & 0 & t_0 & -t_0 & 0 \\
	0 & 0 & 0 & 0 & t_0 & -t_0 \\
	0 & 0 & 0 & 0 & 0 & 0 
	\end{bmatrix} ,\;\;
	\mathbf B_{\mathbf v_{11}}:=\begin{bmatrix}
	0 & 0 & 0 & 0 & 0 & 0 \\
	0 & 0 & 0 & 0 & 0 & 0 \\
	0 & 0 & 0 & 0 & 0 & 0 \\
	0 & 0 & 0 & 0 & 0 & 0 \\
	0 & 0 & 0 & 0 & 0 & 0 \\
	1 & 1 & 1 & t_0 & t_0 & t_0
	\end{bmatrix}.
	\end{equation*}
	These matrices satisfy the three conditions~(j), (jj) and~(jjj) above for $t_0>0$; e.g., $\mathbf A_{\mathbf v_{11}}\mathbf B_{\mathbf v_{11}}^*=0$. For the other vertices of $\mathcal{G}_{2}$, and also for $\mathcal{G}_3$, the arguments are similar. Therefore, the Schr\"odinger operators  $H_{n}$  are self-adjoint. Also note that the condition~(jjj) is equivalent to the conditions \eqref{continuity} and \eqref{vanishingflux} of the weighted vertex conditions presented before.  
\end{remark}	
	
\begin{remark}
Although some quantities depend on the edge potential~$q_0$, as the discriminant~\eqref{eqDiscrimi}, our general results on spectral type characterization and the possible presence of Dirac cones, including their locations (see ahead), do not depend on any specific choice of the admissible~$q_0$; in particular, it can be the null potential. This shows that we have, in fact, a family of models. 
\end{remark}

\section{Spectral analysis}\label{sectionSpectrumGnA}

In this section we use Floquet-Bloch theory \cite{eastham,Floquettheory2K,Floquettheory3RS,Floquettheory4BES} to study the spectrum of $H_{n}$; we extend results of~\cite{KP1}. We begin with general remarks, then we specialize to bilayer and trilayer graphene. 

 For each \textit{quasimomentum} $\theta=(\theta_1,\theta_2)$ in the Brillouin zone $\mathcal{B}=[-\pi,\pi]^{2}$, let $H_{n}(\theta)$ be the \textit{Bloch Hamiltonian} acting in ${\mathrm L}^2(\mathcal{W}_n)$ as in (\ref{defHnA}), but with a different domain: $D(H_{n}(\theta))$ is the subspace of functions~$u$ that satisfy~(i)-(iv) in $D(H_{n})$, and also the following  \textit{Floquet condition}
\begin{equation}\label{floquetcondition}
u(\mathbf x+p_1 {\mathbf E_1} +p_2 {\mathbf E_2})= e^{i\mathbf p\theta}u(\mathbf x)=e^{i(p_1\theta_1 +p_2\theta_2)}u(\mathbf x),
\end{equation}
for all $\mathbf p=(p_1,p_2)\in \mathbb{Z}^2$ and all $\mathbf x\in \mathcal{G}_{n}$. It is well known that $H_{n}(\theta)$ has purely discrete spectrum~\cite{K1}, denoted by $\sigma(H_{n}(\theta))=\{\lambda_k(\theta)\}_{k\geq 1}$.  The ranges of the functions $\theta \mapsto \{\lambda_k(\theta)\}$ is called the \textit{dispersion relation}  of $H_{n}$ and it determines its spectrum  \cite{eastham,Floquettheory2K,Floquettheory3RS,Floquettheory4BES}
\begin{equation}\label{spectrumHnA}
\sigma(H_{n})=\displaystyle\bigcup_{\theta\in\mathcal{B}}\sigma(H_{n}(\theta)).
\end{equation}

The goal now is to determine the spectra of $\sigma(H_{n}(\theta))$, $\theta\in\mathcal{B}$, by solving the eigenvalue problem 
\begin{equation}\label{EVP1}
H_{n}(\theta)u=\lambda u,\quad \lambda \in \mathbb{R},\quad u \in D(H_{n}(\theta)).
\end{equation}
Consider two auxiliary operators. The first one is the  \textit{Dirichlet Schr\"odinger operator} $H^{D}$  that acts in ${\mathrm L}^{2}([0,1])$ as
\begin{equation}
	H^{D}u(x)=-\frac{{\mathrm d}^{2}u(x)}{{\mathrm d}x^2}+q_0(x)u(x),
\end{equation}
where $u$ satisfies the Dirichlet boundary condition, that is,
\begin{equation}\label{dirichletboundarycondition}
u(0)=u(1)=0.
\end{equation}
It is well known that $H^{D}$ has purely discrete spectrum, denoted by $\sigma(H^{D})=\{\lambda^{D}_{k}\}_{k\geq 1}$ (see, for instance, \cite{Floquettheory4BES}). It is worth mentioning that  the Dirichlet conditions correspond to the trivial case of a graph consisting of a union of uncoupled bonds~\cite{kottosSimil}.

To describe the second operator, let $q_p$ be the \textit{potential function} obtained by  extending periodically $q_0$ to the whole real axis $\mathbb{R}$. The \textit{Hill operator} $H^{\mathrm{per}}$ acts in  ${\mathrm L}^{2}(\mathbb{R})$ as
\begin{equation}
H^{\mathrm{per}}u(x)=-\frac{{\mathrm d}^{2}u(x)}{{\mathrm d}x^2}+q_p(x)u(x).
\end{equation}

For the spectral problem 
\begin{equation}\label{EVP for Hp}
H^{\mathrm{per}}\varphi=\lambda\varphi,
\end{equation}
 consider the \textit{monodromy matrix} $\mathbf M(\lambda)$ of $H^{\mathrm{per}}$ given by (see~\cite{Floquettheory4BES})
\begin{equation}\label{definition monodromy matrix}
\begin{bmatrix} \varphi(1) \\ \varphi^{\prime}(1) \end{bmatrix} = \mathbf M(\lambda) \begin{bmatrix} \varphi(0) \\ \varphi^{\prime}(0) \end{bmatrix},
\end{equation}
where $\varphi$ is any solution of the problem~(\ref{EVP for Hp}). The matrix $\mathbf M(\lambda)$ shifts $\begin{bmatrix}\varphi(0) & \varphi^{\prime}(0)\end{bmatrix}^{\intercal}$ by the period of $q_p$ (in our case,~$1$). Let 
\begin{equation}\label{eqDiscrimi}
\mathcal{D}(\lambda):= \mathrm{tr}(\mathbf M(\lambda))
\end{equation} be the \textit{discriminant} of the Hill operator $H^{\mathrm{per}}$. There are many results and properties about the spectrum of the Hill operator $H^{\mathrm{per}}$ and its discriminant~$\mathcal{D}(\lambda)$. For instance, the spectrum $\sigma(H^{\mathrm{per}})$ is purely absolutely continuous and
\begin{equation*}
\sigma(H^{\mathrm{per}})=\{\lambda\in\mathbb{R}:|\mathcal{D}(\lambda)|\leq 2\}. 
\end{equation*}
Furthermore, the spectrum $\sigma(H^{\mathrm{per}})$ is the union of closed  intervals $B_k$, called \textit{bands of} $\sigma(H^{\mathrm{per}})$, in such way that, for $\lambda\in B_k$, $\mathcal{D}^{\prime}(\lambda)\neq 0$ and $\mathcal{D}(\lambda):B_k \longrightarrow [-2,2]$ is a homeomorphism, for each $k$. In particular, in the free case, i.e., when $q_0=0$, $\mathcal{D}(\lambda)=2\cos\sqrt{\lambda}$ (see \cite{KP1}, Proposition 3.4, and also \cite{eastham,Floquettheory2K,Floquettheory3RS,Floquettheory4BES,HillMW}).

\subsubsection*{Bilayer graphene}

The structure that represents the AB-stacked bilayer graphene $\mathcal{G}_{2}$ consists of two graphene sheets, as described in Section~\ref{nAstackednlayergraphene}; see Figure \ref{FundamentalDomainBilayer}. Let us write out the conditions (\ref{continuity}), (\ref{vanishingflux}) and (\ref{floquetcondition}) on the fundamental domain $\mathcal{W}_{2}$ (Figure \ref{FundamentalDomainBilayer}). As we identify each edge~$\mathbf e$ with the interval $[0,1]$, it follows that $\mathbf v_{i1}\sim 0$ and $\mathbf v_{i2}\sim 1$, for $i=1,2$. By condition~(\ref{continuity}),  
\begin{equation}\label{new continuity 0}
\begin{cases}
u_{\mathbf{a}_{11}}(0)=u_{\mathbf{a}_{12}}(0)=u_{\mathbf{a}_{13}}(0)=u_{\mathbf{f}_{1}}(0)/t_0=:\alpha_1 \\
u_{\mathbf{a}_{21}}(0)=u_{\mathbf{a}_{22}}(0)=u_{\mathbf{a}_{23}}(0)=:\alpha_2 
\end{cases}.
\end{equation}
The  Floquet condition (\ref{floquetcondition}) implies that
\begin{equation}
u_{\mathbf{a}_{i1}}(1)=e^{i\theta_1}u_{\mathbf{a}_{i2}}(1) \quad\text{and}\quad u_{\mathbf{a}_{i1}}(1)=e^{i\theta_2}u_{\mathbf{a}_{i3}}(1), \quad i=1,2, 
\end{equation}
so by (\ref{continuity}),
\begin{equation}\label{new continuity 1}
\begin{cases}
u_{\mathbf{a}_{11}}(1)=e^{i\theta_1}u_{\mathbf{a}_{12}}(1)=e^{i\theta_2}u_{\mathbf{a}_{13}}(1)=:\beta_1 \\
u_{\mathbf{a}_{21}}(1)=e^{i\theta_1}u_{\mathbf{a}_{22}}(1)=e^{i\theta_2}u_{\mathbf{a}_{23}}(1)=u_{\mathbf{f}_1}(1)/t_0=:\beta_2 
\end{cases}.
\end{equation}
Similarly, by  condition~(\ref{vanishingflux}),  
\begin{equation}\label{newzeroflux}
\begin{cases}

u_{\mathbf{a}_{11}}^{\prime}(0)+u_{\mathbf{a}_{12}}^{\prime}(0)+u_{\mathbf{a}_{13}}^{\prime}(0)+t_{0}\,u_{\mathbf{f}_{1}}^{\prime}(0)=0 \\
u_{\mathbf{a}_{11}}^{\prime}(1)+e^{i\theta_1}u_{\mathbf{a}_{12}}^{\prime}(1)+e^{i\theta_2}u_{\mathbf{a}_{13}}^{\prime}(1)=0 \\
u_{\mathbf{a}_{21}}^{\prime}(0)+u_{\mathbf{a}_{22}}^{\prime}(0)+u_{\mathbf{a}_{23}}^{\prime}(0)=0 \\
u_{\mathbf{a}_{21}}^{\prime}(1)+e^{i\theta_1}u_{\mathbf{a}_{22}}^{\prime}(1)+e^{i\theta_2}u_{\mathbf{a}_{23}}^{\prime}(1)+t_{0}\,u_{\mathbf{f}_{1}}^{\prime}(1)=0 

\end{cases}.
\end{equation}

Let $\lambda \notin \sigma(H^{D})$. Then there exists two linearly independent solutions $\varphi_{\lambda,0},\varphi_{\lambda,1}$ of the problem 
\begin{equation}\label{dirichletEigenvalueProblem}
-\frac{{\mathrm d}^2\varphi(x)}{{\mathrm d}x^2}+q(x)\varphi(x)=\lambda \varphi(x),
\end{equation}
such that
\begin{equation}\label{conitions on vaphi}
\begin{cases}
\varphi_{\lambda,0}(0)=1\\
\varphi_{\lambda,0}(1)=0
\end{cases}
\quad\quad
\begin{cases}
\varphi_{\lambda,1}(0)=0\\
\varphi_{\lambda,1}(1)=1
\end{cases}
\end{equation}
and  
\begin{equation}\label{conditions on derivative}
\varphi_{\lambda,1}^{\prime}(x)=-\varphi_{\lambda,0}^{\prime}(1-x), \quad x\in[0,1].
\end{equation}

Since each edge of $\mathcal{W}_{2}$ is identified with the interval $[0,1]$, we can define $\varphi_{\lambda,i}$ in each edge and we will keep the same notation  $\varphi_{\lambda,i}$ for such functions. Hence, for each $\lambda\notin\sigma(H^{D})$, we can represent
\begin{equation}\label{representationoffunctionsABonW}
\begin{cases}
u_{\mathbf{a}_{i1}}=\alpha_{i}\varphi_{\lambda,0}+\beta_{i}\varphi_{\lambda,1}\\
u_{\mathbf{a}_{i2}}=\alpha_{i}\varphi_{\lambda,0}+e^{-i\theta_{1}}\beta_{i}\varphi_{\lambda,1}\\
u_{\mathbf{a}_{i3}}=\alpha_{i}\varphi_{\lambda,0}+e^{-i\theta_2}\beta_{i}\varphi_{\lambda,1}\\
u_{\mathbf{f}_{1}}=t_0\alpha_1\varphi_{\lambda,0}+t_0\beta_2\varphi_{\lambda,1}
\end{cases} \quad i=1,2.
\end{equation}
It is easy to see that the function defined by \eqref{representationoffunctionsABonW} satisfies the  conditions (\ref{new continuity 0}) and (\ref{new continuity 1}) and solves the eigenvalue problem \eqref{EVP1}. It remains to verify condition~(\ref{newzeroflux}). By substituting (\ref{representationoffunctionsABonW}) into (\ref{newzeroflux}), 
\begin{equation}\label{system part 1}
\begin{cases}
(3+t_0^2)\alpha_1\varphi_{\lambda,0}^{\prime}(0)+\bar{F}(\theta)\beta_1\varphi_{\lambda,1}^{\prime}(0)+t_0^2\, \beta_2\varphi_{\lambda,1}^{\prime}(0)=0\\
F(\theta)\alpha_1\varphi_{\lambda,0}^{\prime}(1)+3\beta_1\varphi_{\lambda,1}^{\prime}(1)=0\\
3\varphi_{\lambda,0}^{\prime}(0)\alpha_2+\bar{F}(\theta)\beta_2\varphi_{\lambda,1}^{\prime}(0)=0\\
t_0^2 \alpha_1 \varphi_{\lambda,0}^{\prime}(1)+F(\theta) \alpha_2 \varphi_{\lambda,0}^{\prime}(1)+(3+t_0^2)\varphi_{\lambda,1}^{\prime}(1)\beta_2=0
\end{cases},
\end{equation} 
where $F(\theta):=1+e^{i\theta_{1}}+e^{i\theta_2}$ and $\bar{F}(\theta)$ is its complex conjugate. By (\ref{conditions on derivative}), we have that $\varphi_{\lambda,0}^{\prime}(1)=-\varphi_{\lambda,1}^{\prime}(0)$ and $\varphi_{\lambda,0}^{\prime}(0)=-\varphi_{\lambda,1}^{\prime}(1)$. Thus the system (\ref{system part 1}) is equivalent to
\begin{equation}\label{system part 2}
\begin{cases}
-(3+t_0^2)\alpha_1\varphi_{\lambda,1}^{\prime}(1)+\bar{F}(\theta)\beta_1\varphi_{\lambda,1}^{\prime}(0)+t_0^2\, \beta_2\varphi_{\lambda,1}^{\prime}(0)=0\\
-F(\theta)\alpha_1\varphi_{\lambda,1}^{\prime}(0)+3\beta_1\varphi_{\lambda,1}^{\prime}(1)=0\\
-3\varphi_{\lambda,1}^{\prime}(1)\alpha_2+\bar{F}(\theta)\beta_2\varphi_{\lambda,1}^{\prime}(0)=0\\
-t_0^2 \alpha_1 \varphi_{\lambda,1}^{\prime}(0)-F(\theta) \alpha_2 \varphi_{\lambda,1}^{\prime}(0)+(3+t_0^2)\varphi_{\lambda,1}^{\prime}(1)\beta_2=0
\end{cases}.
\end{equation} 
Since $\varphi_{\lambda,1}^{\prime}(0)\neq 0$, the quotient 
\begin{eqnarray}\label{definition of eta}
\eta(\lambda):=\frac{\varphi_{\lambda,1}^{\prime}(1)}{\varphi_{\lambda,1}^{\prime}(0)}
\end{eqnarray}
is well defined. Hence, dividing the system (\ref{system part 2}) by $\varphi_{\lambda,1}^{\prime}(0)$ and multiplying the second and fourth lines by $-1$, we obtain  
\begin{equation}\label{system part 3}
\begin{cases}
-T_0\eta \alpha_1 + \bar{F}\beta_1+t_0^2\beta_2=0\\
F\alpha_1-3\eta\beta_1=0\\
-3\eta\alpha_2+\bar{F}\beta_2=0\\
t_0^2\alpha_1+F\alpha_2-T_0\beta_2=0
 \end{cases},
\end{equation}
where $T_0=3+t_0^2$, $\eta=\eta(\lambda)$ and $F=F(\theta)$. The matrix form of the system~(\ref{system part 3}) is
\begin{equation}
\mathbf{M}_{2}(\eta(\lambda),\theta)X = 0,
\end{equation}
where $X=\begin{bmatrix}\alpha_1 & \beta_1 & \alpha_2 & \beta_2\end{bmatrix}^{\intercal}$ and 

\begin{equation}\label{matrix for AB n=2}
\mathbf{M}_{2}(\eta(\lambda),\theta)=
\begin{bmatrix}
-T_0\eta&\bar{F}&0&t_0^2\\
F&-3\eta&0&0\\
0&0&-3\eta&\bar{F}\\
t_0^2&0&F&-T_0\eta
\end{bmatrix}.
\end{equation}

Note that $\det(\mathbf{M}_{2}(\lambda))$ is a quartic polynomial in $\eta(\lambda)$. Hence, if there exists a $\theta\in\mathcal{B}$ such that 
\begin{equation}\label{det=0}
\det(\mathbf{M}_{2}(\eta(\lambda),\theta))=0,
\end{equation}
that is, $\eta(\lambda)=r(\theta)$, where $r(\theta)$ is one of the four roots of \eqref{det=0}, it follows that the representation (\ref{representationoffunctionsABonW}) solves the eigenvalue problem (\ref{EVP1}) and so, by \eqref{spectrumHnA}, $\lambda\in\sigma(H_{2})$.

\begin{remark}\label{Remark_Roots_and_Etas}
	The four roots $r(\theta)$ of \eqref{det=0} are $\theta$-dependent and does not depend on $\lambda$. Thus, if there exists a $\theta\in\mathcal{B}$ such that \eqref{det=0} holds, then we get a relation between $\eta(\lambda)$ and the function $r(\theta)$. This fact gives explicitly the \textit{dispersion relation}  of $H_{2}$ (and similarly for  $H_3$).
\end{remark} 

\subsubsection*{Trilayer  graphene}
We now analyze the  case $n=3$;  $\mathcal{G}_{3}$ consists of three graphene sheets $G_1, G_2$ and $G_3$. Similarly to the case $n=2$, lets write out the conditions (\ref{continuity}), (\ref{vanishingflux}) and (\ref{floquetcondition}) on the fundamental domain $\mathcal{W}_{3}$ (see Figure \ref{FundamentalDomainTrilayer}). Since
\begin{equation*}
u_{\mathbf{a}_{i1}}(1)=e^{i\theta_1}u_{\mathbf{a}_{i2}}(1) \quad\text{and}\quad u_{\mathbf{a}_{i1}}(1)=e^{i\theta_2}u_{\mathbf{a}_{i3}}(1), \quad i=1,2,3, 
\end{equation*}
it follows that the continuity condition (\ref{continuity}) is equivalent to
\begin{equation}\label{continuity trilayer on W}
\begin{cases}
u_{\mathbf{a}_{11}}(0)=u_{\mathbf{a}_{12}}(0)=u_{\mathbf{a}_{13}}(0)=u_{\mathbf{f}_{1}}(0)/t_0=:\alpha_1\\
u_{\mathbf{a}_{11}}(1)=e^{i\theta_1}u_{\mathbf{a}_{12}}(1)=e^{i\theta_2}u_{\mathbf{a}_{13}}(1)=:\beta_1\\
u_{\mathbf{a}_{21}}(0)=u_{\mathbf{a}_{22}}(0)=u_{\mathbf{a}_{23}}(0)=:\alpha_2\\
u_{\mathbf{a}_{21}}(1)=e^{i\theta_1}u_{\mathbf{a}_{22}}(1)=e^{i\theta_2}u_{\mathbf{a}_{23}}(1)=u_{\mathbf{f}_{1}}(1)/t_0=u_{\mathbf{f}_{2}}(1)/t_0=:\beta_2\\
u_{\mathbf{a}_{31}}(0)=u_{\mathbf{a}_{32}}(0)=u_{\mathbf{a}_{33}}(0)=u_{\mathbf{f}_{2}}(0)/t_0=:\alpha_3\\
u_{\mathbf{a}_{31}}(1)=e^{i\theta_1}u_{\mathbf{a}_{32}}(1)=e^{i\theta_2}u_{\mathbf{a}_{33}}(1)=:\beta_3
\end{cases}
\end{equation}
and the condition (\ref{vanishingflux}) is equivalent to
\begin{equation}\label{zeroflux trilayer on W}
\begin{cases}
u_{\mathbf{a}_{11}}^\prime(0)+u_{\mathbf{a}_{12}}^\prime(0)+u_{\mathbf{a}_{13}}^\prime(0)+t_0\,u_{\mathbf{f}_{1}}^\prime(0)=0\\
u_{\mathbf{a}_{11}}^\prime(1)+e^{i\theta_1}u_{\mathbf{a}_{12}}^\prime(1)+e^{i\theta_2}u_{\mathbf{a}_{13}}^\prime(1)=0\\
u_{\mathbf{a}_{21}}^\prime(0)+u_{\mathbf{a}_{22}}^\prime(0)+u_{\mathbf{a}_{23}}^\prime(0)=0\\
u_{\mathbf{a}_{21}}^\prime(1)+e^{i\theta_1}u_{\mathbf{a}_{22}}^\prime(1)+e^{i\theta_2}u_{\mathbf{a}_{23}}^\prime(1)+t_0\,u_{\mathbf{f}_{1}}^\prime(1)+t_0\,u_{\mathbf{f}_{2}}^\prime(1)=0\\
u_{\mathbf{a}_{31}}^\prime(0)+u_{\mathbf{a}_{32}}^\prime(0)+u_{\mathbf{a}_{33}}^\prime(0)+t_0\,u_{\mathbf{f}_{2}}^\prime(0)=0\\
u_{\mathbf{a}_{31}}^\prime(1)+e^{i\theta_1}u_{\mathbf{a}_{32}}^\prime(1)+e^{i\theta_2}u_{\mathbf{a}_{33}}^\prime(1)=0
\end{cases}.
\end{equation}
Let $\lambda \notin \sigma(H^{D})$ and let $\varphi_{\lambda,0}$ and $\varphi_{\lambda,1}$ be the two linearly independent solutions of the problem (\ref{dirichletEigenvalueProblem}) that satisfies (\ref{conitions on vaphi}) and (\ref{conditions on derivative}). If we write
\begin{equation}\label{representation of functions on W n=3}
\begin{cases}
u_{\mathbf{a}_{i1}}=\alpha_{i}\varphi_{\lambda,0}+\beta_{i}\varphi_{\lambda,1}\\
u_{\mathbf{a}_{i2}}=\alpha_{i}\varphi_{\lambda,0}+e^{-i\theta_{1}}\beta_{i}\varphi_{\lambda,1}\\
u_{\mathbf{a}_{i3}}=\alpha_{i}\varphi_{\lambda,0}+e^{-i\theta_2}\beta_{i}\varphi_{\lambda,1}\\
u_{\mathbf{f}_{1}}=t_0\alpha_1\varphi_{\lambda,0}+t_0\beta_2\varphi_{\lambda,1}\\
u_{\mathbf{f}_{2}}=t_0\alpha_3\varphi_{\lambda,0}+t_0\beta_2\varphi_{\lambda,1}
\end{cases} \quad i=1,2,3,
\end{equation}
the continuity condition (\ref{continuity trilayer on W}), as well  problem (\ref{EVP1}), are satisfied. It remains to verify the condition~(\ref{zeroflux trilayer on W}). The matrix form of the obtained system  in this case is
\begin{equation*}
\mathbf{M}_{3}(\eta(\lambda),\theta)\mathbf X=0, 
\end{equation*}
with $\mathbf X=\begin{bmatrix}\alpha_1 & \beta_1 & \alpha_2 & \beta_2 & \alpha_3 & \beta_3\end{bmatrix}^{\intercal}$ and
\begin{equation}\label{matrix for AB n=3}
\mathbf{M}_{3}(\eta(\lambda),\theta)=
\begin{bmatrix}
\mathbf M_{T_0\tilde{T}_0} & \tilde{\mathbf m}_{t_0}^{\intercal}\\
\tilde{\mathbf m}_{t_0} & {\mathbf N}_{T_0} 
\end{bmatrix},
\end{equation}
where $T_0=3+t_0^2$, $\tilde{T}_0=3+2t_0^2$, $F=F(\theta)=1+e^{i\theta_1}+e^{i\theta_2}$, 
\begin{equation}\label{matrix M_T1T2}
\mathbf M_{T_0\tilde{T}_0}=
\begin{bmatrix}
-T_0\eta&\bar{F}&0&t_0^2\\
F&-3\eta&0&0\\
0&0&-3\eta&\bar{F}\\
t_0^2&0&F&-\tilde{T}_0\eta
\end{bmatrix}, 
\end{equation}
\begin{equation} \label{matrix N_T0}
\tilde{\mathbf m}_{t_0}=
\begin{bmatrix}
0 & 0 & 0 & t_0^2\\
0 & 0 & 0 & 0 
\end{bmatrix}
\quad \text{and} \quad
{\mathbf N}_{T_0}=
\begin{bmatrix}
-T_0\eta & \bar{F}\\
F & -3\eta
\end{bmatrix}.
\end{equation}

Note that $\mathbf{M}_{2}(\eta(\lambda),\theta)=\mathbf M_{T_0T_0}$. Thus, if there exists $\theta\in\mathcal{B}$ such that $\det(\mathbf{M}_{3}(\eta(\lambda),\theta))=0$, the representation (\ref{representation of functions on W n=3}) solves the eigenvalue problem~(\ref{EVP1}) and so $\lambda\in\sigma(H_{3})$.

\subsubsection*{Joint spectral analysis}
Therefore, we have the following result:

\begin{prop}\label{result det(M,n,AB)=0}
Let $\lambda\notin\sigma(H^{D})$. Then, for $n=2,3$ and $t_0>0$, the real number $\lambda \in \sigma(H_{n})$ if and only if there exists $\theta\in\mathcal{B}$ such that
\begin{equation}\label{eqdetMn}
\det\left(\mathbf{M}_{n}(\eta(\lambda),\theta)\right)=0\,.
\end{equation}
\end{prop}

Note that $\det(\mathbf{M}_{n}(\eta(\lambda),\theta))$ is a polynomial of degree $2n$ in $\eta(\lambda)$, for $n=2,3$. By Proposition \ref{result det(M,n,AB)=0}, the spectra $\sigma(H_{n})$ are basically determined if we know the range of all $2n$ roots $r(\theta)$ of $\det(\mathbf{M}_{n}(\eta(\lambda),\theta))$ as a polynomial in $\eta(\lambda)$ (see Remark \ref{Remark_Roots_and_Etas})). For the case $n=2$ we can easily calculate the roots of~\eqref{eqdetMn}. We use the \textit{Laplace's expansion formula} for matrix determinants in $\mathbf{M}_{2}(\eta(\lambda),\theta)$, given by (\ref{matrix for AB n=2}), to obtain
\begin{eqnarray}
\label{QuarticEquation}\det(\mathbf{M}_{2}(\eta(\lambda),\theta))&=&9T_0^{2}\eta^{4}(\lambda)-(9t_0^4+6T_0F(\theta)\bar{F}(\theta))\eta^{2}(\lambda)\\ 
&+&\left(F(\theta)\bar{F}(\theta)\right)^2=0. \nonumber
\end{eqnarray}
Note that we can easily turn the quartic equation \eqref{QuarticEquation} into a quadratic one. 
Thus the four roots are
\begin{equation}\label{roots of delta polynomial of AB}
r_{\pm}^\pm(\theta)=\pm\sqrt{\frac{G_2(t_0,\theta)\pm\sqrt{G_2(t_0,\theta)^2-36T_0^2(F\bar{F})^2}}{18T_0^2}},
\end{equation}
where $G_2(t_0,\theta)=9t_0^4+6T_0F\bar{F}$. Here, the subscript $\pm$ refers to the outside $``\pm"$ of the first square root symbol while the superscript~$\pm$ refers to the inside one. 

For  $n=3$, applying Laplace's formula to~(\ref{matrix for AB n=3}), we obtain 
\begin{eqnarray}
\label{detM3AB=0 explicito}\det(\mathbf{M}_{3}(\eta(\lambda),\theta)) &=& 27T_0^2\tilde{T}_0\eta^6 -(54T_0t_0^4+9(T_0^2+2T_0\tilde{T}_0)F\bar{F})\eta^4\\
&+& (18t_0^4F\bar{F}+3(2T_0+\tilde{T}_0)(F\bar{F})^2)\eta^2 -(F\bar{F})^3.\nonumber
\end{eqnarray}
The six roots of $\det(\mathbf{M}_{3}(\eta(\lambda),\theta))$ are the following:
\begin{equation}\label{roots for n=3 part 1}
\tilde{r}_{\pm}^{\pm}(\theta)=\pm\sqrt{\frac{G_3(t_0,\theta)\pm\sqrt{G_3(t_0,\theta)^2-36T_0\tilde{T}_0(F\bar{F})^2}}{18T_0\tilde{T}_0}},
\end{equation}
\begin{equation}\label{roots for n=3 part 2}
\bar{r}_{\pm}(\theta)=\pm\sqrt{\frac{F\bar{F}}{3T_0}}, 
\end{equation}
where $G_3(t_0,\theta)= 18t_0^4 +3(T_0+\tilde{T}_0)F\bar{F}$.

\begin{remark}\label{remarkRoots3}
To calculate the exact roots \eqref{roots for n=3 part 1} and \eqref{roots for n=3 part 2} of~$\det(\mathbf{M}_{3}(\eta(\lambda),\theta))$, we have first replaced, in  $\mathbf{M}_{3}(\eta(\lambda),\theta)$, the matrices $\tilde{\mathbf m}_{t_0}$ and $\tilde{\mathbf m}_{t_0}^\intercal$ by zero, to obtain a block matrix $\tilde{\mathbf{M}}_{3}(\eta(\lambda),\theta)$, so that
\begin{equation*}
\det(\tilde{\mathbf{M}}_{3}(\eta(\lambda),\theta))=\det(\mathbf M_{T_0\tilde{T}_0})\det ({\mathbf N}_{T_0}). 
\end{equation*}
Hence, the six roots of $\det(\tilde{\mathbf{M}}_{3}(\eta(\lambda),\theta))$ are composed by the four roots of $\det(\mathbf M_{T_0\tilde{T}_0})$ and the two roots of $\det ({\mathbf N}_{T_0})$. Since the two determinants
\begin{equation}
\det (\mathbf{M}_{3}(\eta(\lambda),\theta))\quad \text{and} \quad  \det(\tilde{\mathbf{M}}_{3}(\eta(\lambda),\theta))
\end{equation}
have similar expressions, it was observed that the six roots of the original determinant should be similar to the roots of  $\det(\tilde{\mathbf{M}}_{3}(\eta(\lambda),\theta))$. Then, \eqref{roots for n=3 part 1} was obtained from inspection and suitable modifications on the roots of $\det (\mathbf M_{T_0\tilde{T}_0})$; \eqref{roots for n=3 part 2} are the roots of $\det ({\mathbf N}_{T_0})$ (the latter resembles the single layer case~\cite{KP1}).
\end{remark}

Now we just check that one is able to extend some results of~\cite{KP1} in order to characterize the spectrum of~$H_{n}$. For all that, first we discuss details of the relation between the function $\eta(\lambda)$ and the discriminant $\mathcal{D}(\lambda)$ of the Hill operator $H^{\mathrm{per}}$, for $n=2,3$. This is important to relate the spectrum of  $H^{\mathrm{per}}$ to the spectrum of~$H_{n}$. 

\begin{lem} (\cite{KP1}, page 813)\label{lemma relaton D() and eta()}
	Let $\lambda\notin\sigma(H^D)$ and  $\mathcal{D}(\lambda)$ be the discriminant of the Hill operator $H^{\mathrm{per}}$. Then for $H_2$ and $H_3$, we have
	\begin{equation}\label{relation D(lambda) and eta(lambda)}
	\eta(\lambda)=\frac{1}{2}\mathcal{D}(\lambda).
	\end{equation}
\end{lem}

Let $\lambda\in\sigma(H^D)$, i.e., an eigenvalue of the Dirichlet operator~$H^D$. For $n=2,3$,  the proof of Lemma~\ref{lemmaLambdaSigma(Hd)} is analogous to the proof  Lemma~3.5 of~\cite{KP1} and Lemma~6 of~\cite{DoKuch}.  

\begin{lem}\label{lemmaLambdaSigma(Hd)}
	Each $\lambda\in\sigma(H^D)$ is an eigenvalue of infinite multiplicity of~$H_n$. 
	\end{lem}

We are ready to characterize the spectra. The following result is proved with the same arguments presented in~\cite{DoKuch}, Theorem~7, and~\cite{KP1}, Theorem~3.6. 

\begin{teo}\label{TheoChacSpectrum}	For the Bernal-stacked  graphene Schr\"odinger operators  $H_n$, $n=2,3$, we have: 
	\begin{itemize}
		\item[(i)] The singular continuous spectrum of the  $H_n$ is empty.
		\item[(ii)] The dispersion relation of $H_n$ consists of two parts:
		\begin{itemize}
			\item[$\bullet$] the pairs $(\lambda,\theta)$ ($\lambda\notin\sigma(H^D)$) such that 
			\begin{equation}\label{dispersionRelationFinalVersion}
			\mathcal{D}(\lambda)=2r(\theta),
			\end{equation}
			where $r(\theta)$ (which depends on~$t_0$ and~$n$) are the~$2n$ solutions of the equation $\det(\mathbf{M}_{n}((\eta(\lambda),\theta))=0$;
			\item[$\bullet$] the collection of flat branches $\lambda\in\sigma(H^D)$, that is, the pairs $(\lambda,\theta)$ for any $\theta\in\mathcal{B}$.
		\end{itemize}
		\item[(iii)] Its absolutely continuous spectrum $\sigma_{\mathrm{ac}}(H_n)$ coincides, as a set, with $\sigma(H^{\mathrm{per}})$, that is, it has a band-gap structure and
		\begin{equation}\label{abs cont spectrum of H^AB}
		\sigma_{\mathrm{ac}}(H_n)=\left\{\lambda\in\mathbb{R}:|\mathcal{D}(\lambda)|\leq 2\right\}.
		\end{equation}
		\item[(iv)] The pure point spectrum of $H_n$ coincides with $\sigma(H^D)$, and each  $\lambda\in\sigma(H^D)$ is an eigenvalue of infinite multiplicity of~$H_n$.
	\end{itemize}
\end{teo}

For the bilayer graphene operator $H_2$, the four curves  of the dispersion relation (\ref{dispersionRelationFinalVersion}) reduces to
\begin{eqnarray*}
\mathcal{D}(\lambda)=\pm2\sqrt{\frac{G_2(t_0,\theta)\pm\sqrt{G_2(t_0,\theta)^2-36T_0^2(F\bar{F})^2}}{18T_0^2}},
\end{eqnarray*}
where $G_2(t_0,\theta)=9t_0^4+6T_0F\bar{F}$, and to obtain the six curves of the dispersion relation of the trilayer graphene operator $H_{3}$, (\ref{dispersionRelationFinalVersion}) takes the form (recall~\eqref{roots for n=3 part 1} and~\eqref{roots for n=3 part 2})
\begin{eqnarray*}
\mathcal{D}(\lambda)&=&\pm2\sqrt{\frac{G_3(t_0,\theta)\pm\sqrt{G_3(t_0,\theta)^2-36T_0\tilde{T}_0(F\bar{F})^2}}{18T_0\tilde{T}_0}} 
\end{eqnarray*}
and
\begin{eqnarray*}
\mathcal{D}(\lambda)&=& \pm2\sqrt{\frac{F\bar{F}}{3T_0}},
\end{eqnarray*}
where $G_3(t_0,\theta)= 18t_0^4 +3(T_0+\tilde{T}_0)F\bar{F}$.

\section{Dirac cones}\label{sectionDiracCones}

 As mentioned in the Introduction, roughly, a \textit{Dirac cone} is a point where two spectral bands linearly touch each other, at least in lowest order approximation, and the quasimomentum $\theta_D\in\mathcal{B}$ for which a Dirac cone occurs is called a \textit{D-point}. It is convenient to present a definition. We say that $\theta_D\in \mathcal B$ is a D-point, if there is a constant $\gamma\ne0$ so that  
\begin{equation}\label{eqDefDiracCone}
\lambda(\theta)-\lambda(\theta_D)+\mathcal{O}((\lambda(\theta)-\lambda(\theta_D))^2)=\pm \gamma|\theta-\theta_D|+\mathcal{O}(|\theta-\theta_D|^2),
\end{equation}
with the~``-''  and~``+'' signs for the valence and conducting bands, respectively.

\begin{teo} \label{mainresult}
For  the Bernal-stacked bilayer and trilayer graphene Schr\"odinger operators $H_2$ and $H_3$, respectively,   and each $t_0>0$, we have:
	\begin{enumerate}
		\item [(i)] The dispersion relation of  $H_2$ presents  parabolic touches at the points $\pm(2\pi/3,-2\pi/3)$ in the Brillouin zone, but no Dirac cone (see Figures \ref{DispersionRelationBilayer} and~\ref{DispersionRelationBilayer3D}).
			\item[(ii)] The dispersion relation of $H_3$ presents  (see Figures~\ref{DispersionRelationTrilayer} and~\ref{DispersionRelationTrilayer3D}):
		\begin{itemize}
			\item two Dirac cones  at the D-points $\pm(2\pi/3,-2\pi/3)$, and
			\item two parabolic touches at the same points $\pm(2\pi/3,-2\pi/3)$. 
		\end{itemize}

	\end{enumerate}
\end{teo}

\begin{proof}
	(i) By Theorem \ref{TheoChacSpectrum}(ii), the non-constant part of the dispersion relation for the AB-stacked bilayer graphene is given by $\mathcal{D}(\lambda)=2r_{\pm}^{\pm}(\theta)$, where
	\begin{equation}\label{RootsForBilayer}
	r_{\pm}^{\pm}(\theta)=\pm\sqrt{\frac{G_2\pm\sqrt{G_2^2-36T_0^2(F\bar{F})^2}}{18T_0^2}},
	\end{equation}
	with $G_2=G_2(t_0,\theta)=9t_0^4+6T_0F\bar{F}$ and $F=F(\theta)=1+e^{i\theta_1}+e^{i\theta_2}$.
	It is easy to show that
	
	\begin{lem} \label{lema dos max e min} The four roots $r_\pm^\pm(\theta)$ of $\det(\mathbf{M}_2(\eta(\lambda),\theta))$ satisfies: 
		\begin{itemize}
			\item [(i)] $\max r_{+}^{+}(\theta)=1$ and $\min r_{+}^{+}(\theta)=t_0^2/(3+t_0^2)$, which are attained at $(0,0)$ and $\pm(2\pi/3,-2\pi/3)$, respectively. 
			\item [(ii)] $\max r_{+}^{-}(\theta)=3/(3+t_0^2)$ and $\min r_{+}^{-}(\theta)=0$, which are attained at $(0,0)$ and $\pm(2\pi/3,-2\pi/3)$, respectively.
			\item [(iii)] $\max r_{-}^{+}(\theta)=-t_0^2/(3+t_0^2)$ and $\min r_{-}^{+}(\theta)=-1$, which are attained at $\pm(2\pi/3,-2\pi/3)$ and $(0,0)$, respectively.  
			\item [(iv)] $\max r_{-}^{-}(\theta)=0$ and $\min r_{-}^{-}(\theta)=-3/(3+t_0^2)$, which are attained at $\pm(2\pi/3,-2\pi/3)$ and $(0,0)$, respectively. 								
		\end{itemize}
	\end{lem}

	\begin{figure}[h]
		
		\centering 
		\includegraphics[width=8cm]{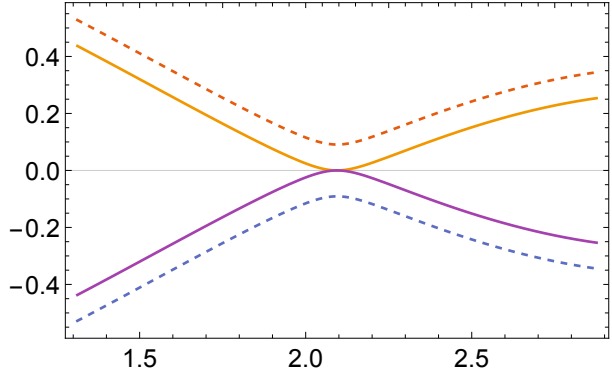}
		\caption{The dispersion relation of the AB-stacked bilayer graphene operator restricted to $\mathcal{B}_d$ and $\theta_1\in\left[\frac{2\pi}{3}-\frac{\pi}{4},\frac{2\pi}{3}+\frac{\pi}{4}\right]$, with the parameter $t_0=0.55$ (and $\theta_2=-\theta_1$). There are no Dirac cones, only quadratic touches.}
		\label{DispersionRelationBilayer}

	\end{figure}
	
	\begin{figure}[h]
		
		\centering 
		\includegraphics[width=8cm]{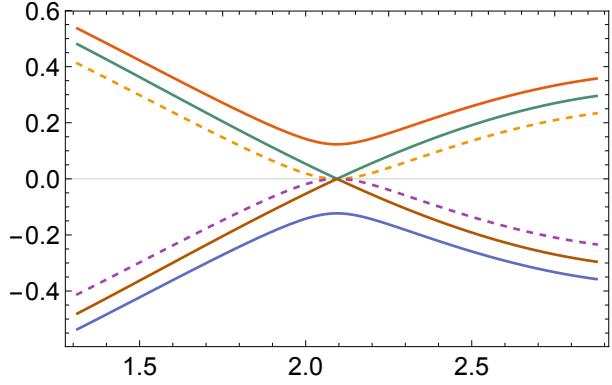}
		\caption{The dispersion relation of the AB-stacked trilayer graphene operator restricted to $\mathcal{B}_d$ and $\theta_1\in\left[\frac{2\pi}{3}-\frac{\pi}{4},\frac{2\pi}{3}+\frac{\pi}{4}\right]$, with the parameter $t_0=0.55$ (and $\theta_2=-\theta_1$). There is a Dirac cone at $2\pi/3$. Note that there is also a quadratic touch at~$2\pi/3$. This qualitatively reproduces, for instance, the experimental results reported in Figure~1(e) of~\cite{BYWCAAZABA2}. }
		\label{DispersionRelationTrilayer}

	\end{figure}

		\begin{figure}[h]
		
		\centering 
		\includegraphics[width=10cm]{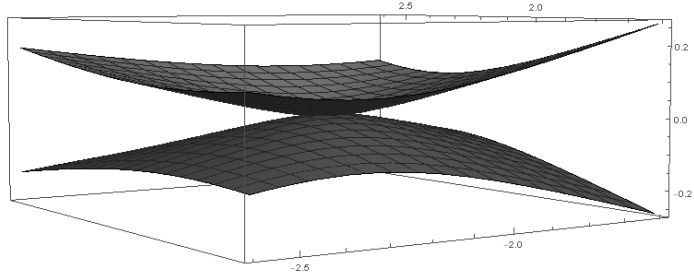}
		\caption{The quadratic touch of the dispersion relation of the AB-stacked bilayer graphene operator restricted to a neighbourhood of the point $(2\pi/3,-2\pi/3)$. The parameter $t_0=0.55$ was considered. The other two surfaces were omitted (see Figure \ref{DispersionRelationBilayer}).}
		\label{DispersionRelationBilayer3D}

	\end{figure}

		\begin{figure}[h]
		
		\centering 
		\includegraphics[width=10cm]{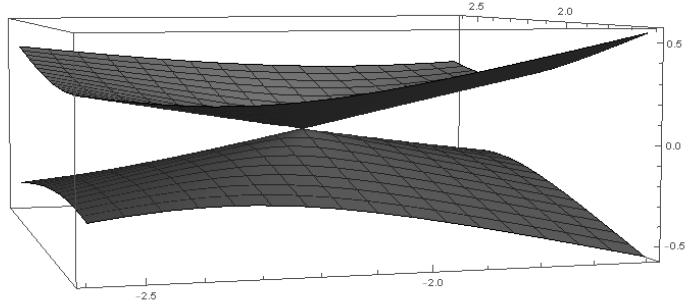}
		\caption{The Dirac cone of the dispersion relation of the AB-stacked trilayer graphene operator restricted to a neighbourhood of the D-point $(2\pi/3,-2\pi/3)$. The parameter $t_0=0.55$ was considered. The other four surfaces where omitted, including the quadratic touch similar to the one in Figure~\ref{DispersionRelationBilayer3D}.}
		\label{DispersionRelationTrilayer3D}
	\end{figure}
	
	Note that the possible Dirac cones in this case depend on the function 
	\begin{equation*}
	F(\theta)\bar{F}(\theta)=|F(\theta)|^2=1+8\cos\left(\frac{\theta_1-\theta_2}{2}\right)\cos\left(\frac{\theta_1}{2}\right)\cos\left(\frac{\theta_2}{2}\right).
	\end{equation*}
	We have that $|F(\theta)|=0$ if and only if $\theta=\theta_\pm=\pm(2\pi/3,-2\pi/3)$, and since these points belong to the diagonal 
	\[
	\mathcal{B}_d:=\{\theta\in\mathcal{B}:\theta_1=-\theta_2\}
	\] of the Brillouin zone $\mathcal{B}$, it suffices to study $r_{\pm}^{\pm}(\theta)$ restricted to $\mathcal{B}_d$. Note the touch at $\theta_\pm=\pm(2\pi/3,-2\pi/3)$ given by Lemma~\ref{lema dos max e min}, items~(ii) and~(iv), where $\min r_{+}^{-}(\theta_\pm)=\max r_{-}^{-}(\theta_\pm)=0$ (that is, $\mathcal D(\lambda)$ vanishes).
	
	For $\theta\in\mathcal{B}_d$, $F(\theta)=1+2\cos(\theta_1)$ and then \eqref{RootsForBilayer} takes the form
	\begin{equation}\label{EtaRestrictDiag}
	r_{\pm}^{\pm}(\theta)=\pm\sqrt{\frac{G_2\pm\sqrt{G_2^2-36T_0^2F^4}}{18T_0^2}}.
	\end{equation} 
	Calculating $dr_{\pm}^{\pm}/d\theta$, we get	
	\begin{equation}\label{DerivativeEtaRestrict}
	\frac{dr_{\pm}^{\pm}(\theta)}{d\theta}=\pm\frac{G_2^\prime\pm\frac{2G_2G_2^\prime-144T_0^2F^3F^\prime}{2(G_2^2-36T_0^2F^4)^{1/2}}}{2(18T_0^2)^{1/2}(G_2\pm(G_2^2-36T_0^2F^4)^{1/2})^{1/2}},
	\end{equation}
	where $G_2^\prime=12T_0FF^\prime$. Thus, $dr_{\pm}^{+}/d\theta$ do exist for every $\theta\in\mathcal{B}_d$ and vanish at $\pm 2\pi/3$, since $F(\pm2\pi/3)=0$;  so the branches $r_{\pm}^{+}$ have a quadratic behaviour close to~$\pm 2\pi/3$ and do not present Dirac cones. However, we have an indetermination (i.e., $0/0$) for $dr_{\pm}^{-}/d\theta$ in $\pm2\pi/3$ (otherwise they are well behaved); the  behaviour of these functions around  $\pm2\pi/3$ can be obtained by recalling that, for small~$x$, we have
	\begin{equation*}
	(1+x)^{1/2}\approx 1+\frac{x}{2}-\frac{x^2}{8}.	
	\end{equation*} 
	Then, for $\theta$ close to $\pm 2\pi/3$, 
	\begin{eqnarray*}
	r_{\pm}^{-}(\theta)&=&\pm\sqrt{\frac{9t_0^4+6T_0F^2-\sqrt{(9t_0^4+6T_0F^2)^2-36T_0^2F^4}}{18T_0^2}}\\
	&=& \pm \sqrt{\frac{9t_0^4+6T_0F^2-9t_0^4(1+\frac{4T_0}{3t_0^4}F^2)^{1/2}}{18T_0^2}}\\
	&\approx& \pm\sqrt{\frac{9t_0^4+6T_0F^2-9t_0^4(1+\frac{2T_0}{3t_0^4}F^2-\frac{2T_0^2}{9t_0^8}F^4)}{18T_0^2}}\\
	&=&\pm\frac{F^2}{3t_0^2}. 
	\end{eqnarray*}
Thus, $r_\pm^{-}(\theta)$ also have quadratic behaviour near~$\pm2\pi/3$. Therefore, the dispersion relation of the AB-stacked bilayer graphene does not present Dirac cones and have two parabolic touches at $\pm(2\pi/3,-2\pi/3)$.

(ii) By Theorem \ref{TheoChacSpectrum}(ii), the non-constant part of the dispersion relation for the AB-stacked trilayer graphene is given by $\mathcal{D}(\lambda)=2r(\theta)$,
where $r(\theta)$ are the six roots of $\det (\mathbf M_{3}(\eta(\lambda),\theta))$, which are given by~\eqref{roots for n=3 part 1} and~\eqref{roots for n=3 part 2}. Analogously to the bilayer graphene case, it is sufficient to consider $\theta\in\mathcal{B}_d$. Thus, \eqref{roots for n=3 part 1} turns to 
\begin{equation}\label{NewRootsPart1}
\tilde{r}_{\pm}^{\pm}(\theta)=\pm\sqrt{\frac{G_3\pm\sqrt{G_3^2-36T_0\tilde{T}_0F^4}}{18T_0\tilde{T}_0}},
\end{equation}
where $G_3=G_3(t_0,\theta)=9t_0^4+3(T_0+\tilde{T}_0)F^2$ and $F=F(\theta)=1+2\cos(\theta_1)$, and \eqref{roots for n=3 part 2} takes the form
\begin{equation}\label{NewRootsPart2}
\bar{r}_{\pm}(\theta)=\pm\frac{|F(\theta)|}{\sqrt{3T_0}}.
\end{equation}
The proofs that \eqref{NewRootsPart1} does not have Dirac points and have two parabolic touches at $\pm(2\pi/3,-2\pi/3)$ are analogous to the proof presented for item~(i). It remains to show that $\mathcal{D}(\lambda)=2\bar{r}_{\pm}(\theta)$ satisfies \eqref{eqDefDiracCone}.

Expanding \eqref{NewRootsPart2} in Taylor's series around $\theta_D:=\pm2\pi/3$, we get
\begin{equation}\label{NewRootsPart2Expanded}
\bar{r}_{\pm}(\theta)-\bar{r}_{\pm}(\theta_D)=\pm\bar{\gamma}|\theta_1-\theta_D|+\mathcal{O}(|\theta_1-\theta_D|^2),
\end{equation}
where $\bar{\gamma}=\sqrt{\frac{1}{T_0}}$, since $\bar{r}_{\pm}(\theta_D)=0$ and 
\begin{equation*}
\cos\theta_1=-\frac{1}{2}\mp\frac{\sqrt{3}}{2}(\theta_1-\theta_D)+\mathcal{O}((\theta_1-\theta_D)^2).
\end{equation*}
We have $\bar\gamma\ne0$, since $t_0>0$. It remains to analyze $\mathcal{D}(\lambda)=\mathcal{D}(\lambda(\theta))$. Since $\mathcal{D}^\prime(\lambda(\theta))\neq0$ in the spectral bands of $\sigma(H_3)$ (see \cite{KP1}, Proposition 3.4), then we can expand $\mathcal{D}(\lambda(\theta))$ in Taylor's series around $\lambda(\theta_D)$, to obtain\\
\begin{equation}\label{ExpandingD}
\mathcal{D}(\lambda(\theta))-\mathcal{D}(\lambda(\theta_D))=\mathcal{D}^\prime(\lambda(\theta_D))(\lambda(\theta)-\lambda(\theta_D))+\mathcal{O}((\lambda(\theta)-\lambda(\theta_D))^2).
\end{equation}
In particular, when $q_0=0$, that is, in the free case, since $\mathcal{D}(\lambda(\theta))=2\cos\sqrt{\lambda(\theta)}$, it follows that
\begin{equation*}
\mathcal{D}^\prime(\lambda(\theta))=-\lambda^\prime(\theta_D)\frac{\sin(\sqrt{\lambda(\theta_D)})}{\sqrt{\lambda(\theta)}}. 
\end{equation*}

Combining \eqref{NewRootsPart2Expanded} and \eqref{ExpandingD}, we obtain \eqref{eqDefDiracCone}, with $\gamma=2\bar{\gamma}/\mathcal{D}^\prime(\lambda(\theta_D))$. Therefore, the dispersion relation of the AB-stacked trilayer graphene have two Dirac cones, precisely at $\pm(2\pi/3,-2\pi/3)$ in the Brillouin zone~$\mathcal{B}$, and the proof of the theorem is complete.  
\end{proof}

\section{Comparison with physics literature}\label{sectComparePL}

Our study of the spectrum and the dispersion relation of the AB-stacked bi- and trilayer graphene is based on a limit model, using periodic quantum graph structures. The advantage of this method is that the dispersion relation has an explicit analytical expression. A shortcoming of this approach is the restriction that all edges must have the same length, and so it is not possible to extend it to include second nearest neighbours, for instance (as it is common in more sophisticated tight-binding models). 

The obtained results for the proposed models are consistent with the physics literature. In summary, AB-stacked bi- and trilayer models  have a gapless band component and thus may be characterized as a semimetal, with   parabolic band touch for the bilayer and  the presence of Dirac cones for the trilayer.  Before we elaborate a little about such comparisons, it is worth mentioning that this opens the possibility of  chiral particles with a parabolic nonrelativistic energy spectrum~\cite{novoselovEtal2006,mccannKosh}; a subject to be mathematically investigated. 

Tight-binding models are standard in the physics literature; the single layer graphene case may be recalled in~\cite{castroEtAl} and for the bilayer see~\cite{RSRFAB1,mccannKosh}. Such models of graphene include second and third nearest neighbours (the first graphene tight-binding model~\cite{wallace} has indeed considered second nearest neighbours). The rigorously obtained gapless parabolic band touches at two points for the AB bilayer, illustrated in Figure~\ref{DispersionRelationBilayer}, are similarly predicted by tight-binding calculations; see Figure~8 in~\cite{RSRFAB1} and Figure~3 in~\cite{mccannKosh}. These features  are important for the description of physical properties of AB bilayer graphene, as optical and transport properties~\cite{mccannKosh}.

Such parabolic touches  in the dispersion relations for AB bilayer graphene are also observed in calculations by density functional theory \cite{mSBMac,latilH}. The same general behaviour was obtained by an effective two-dimensional Hamiltonian~\cite{McCannFalko2006} that acts in a space of two-component wave functions (\`a la Dirac equation).  


 For the trilayer  graphene with AB-stacking, the presence of Dirac cones was observed in \cite{LJXABA1,MMABAandABAB} through tight-binding calculations; see, in particular, Figure~2 in~\cite{LJXABA1} and Figure~4 in~\cite{MMABAandABAB}. This is compatible with Theorem~\ref{mainresult}(ii) for our graph model and illustrated in Figure~\ref{DispersionRelationTrilayer}, which also reproduces the experimental data reported in Figure~1(e) of~\cite{BYWCAAZABA2}.  
 
 Another point of agreement with the physics literature is that, for the trilayer case, the Dirac cones come from  one pair of curves in the dispersion relation, and they are absent in the other two pairs~\cite{MMABAandABAB}; it can be seen as a combination~\cite{castroEtAl} of features of a single graphene sheet with the bilayer one, exactly what our results reveal explicitly (see, e.g., Remark~\ref{remarkRoots3}).
 
For the trilayer graphene, we have found that the existence and location of D-points are independent of the value of the interlayer interaction parameter~$t_0$. It would be interesting to investigate whether this occurs for, say, tight-binding models, at least for a range of interaction parameters; we have not found any result in this direction in the literature. 

It is not expected that we could push the physical comparisons of our results far beyond the band touches, since spectral values of energy, for example, depend on the choice of the edge potential~$q_0$, and we have no clue on the (possible) best choices. Together with the necessity of all edges having the same length, an apparent artifact of the graph approximation is the presence of eigenvalues of infinite multiplicity (from Dirichlet boundary conditions), which has not been observed in other approaches to such problems.

To finish, we (again) underline that although our proposed model could be considered rather simple and have some freedom in the choice of the edge potential, it is possible to carry out explicitly and rigorous calculations that recover important findings, by other methods, from the physics literature.

\subsubsection*{Acknowledgments} 
CRdO thanks Professor A. J. Chiquito for discussions on aspects of graphene, and the partial support by CNPq (a Brazilian government agency, under contract 303503/2018-1). VLR thanks the financial support by CAPES (a Brazilian government agency).

\end{document}